\documentclass[reqno,12pt]{amsart}

\usepackage{enumerate}
\usepackage[margin=1in]{geometry}
\usepackage{enumitem}
\usepackage{ifpdf}
\usepackage{amsmath}
\usepackage{amsfonts}
\usepackage{amssymb}
\usepackage{amsaddr}
\usepackage{amsthm}
\usepackage{amsrefs}
\usepackage{mathrsfs}
\usepackage{cite}
\usepackage[hyperfootnotes=false]{hyperref}
\usepackage[dotinlabels]{titletoc}
\usepackage{nicefrac}
\usepackage{float}
\usepackage[multiple]{footmisc}
\usepackage{graphicx}
\usepackage{dcolumn}
\usepackage{tabu}
\usepackage{mathtools}

\newcommand{\ket}[1]{\ensuremath{\left|#1\right\rangle}}

\newtheorem{thm}{Theorem}[section]

\newtheorem{prop}[thm]{Proposition}

\newtheorem{cor}[thm]{Corollary}

\newtheorem{lemma}[thm]{Lemma}

\newtheorem{scho}[thm]{Scholium}

\theoremstyle{definition}
\newtheorem{defn}[thm]{Definition}

\theoremstyle{remark}
\newtheorem{remark}[thm]{Remark}

\theoremstyle{note}

\newcolumntype{A}{D{.}{.}{2.3}}

\makeatletter

\newcommand{\Rmnum}[1]{\expandafter\@slowromancap\romannumeral #1@}
\makeatother

\author{Asif Shakeel}
\email{asif.shakeel@gmail.com}
\ifpdf
\fi

\title[Equivalence of Schr\"{o}dinger and Heisenberg Pictures in QCA] {The Equivalence of  Schr\"{o}dinger and Heisenberg Pictures in Quantum Cellular Automata}

\begin{document}



\begin{abstract}
Quantum cellular automata (QCA) are discrete models of space and time homogeneous quantum field theories (QFTs) and regarded as natural candidates for quantum simulation. Description of a QCA over the separable  Hilbert space of  finite, unbounded configurations (UFC Hilbert space) with  unitary state evolution is the {\it Schr\"{o}dinger template},  and  over the  incomplete infinite  tensor product algebra (ITPA) with evolution by algebra automorphism is  the {\it Heisenberg template}. Whether every  Heisenberg template admits an  equivalent Schr\"{o}dinger template is a foundational question, and one that has  persisted as an open problem. In the present  paper we prove that for every Heisenberg template  an equivalent  Schr\"{o}dinger template exists. We frame the question from a representation theory standpoint, using constructs and results from the representation theory of finite and countably infinite dimensional vector spaces and from category theory to answer it.  With the previously known existence of a Heisenberg template for every  Schr\"{o}dinger template, our result  establishes the equivalence of  the templates. 
\end{abstract}


\maketitle

\section{Introduction} \label{section:intro}

von Neumann~\cite{bib:vonneu1} conceived of  infinitely many  systems interacting, replicating  and  producing complex patterns of behavior. He called them {\it cellular automata}. Likewise, quantum cellular automata (QCA) are  models of  infinitely many {\it quantum} systems  interacting with each other.  Proposed by Feynman~\cite{f:spwc}, these are  natural  models of  quantum systems that would simulate other  quantum systems. Theory of QCA emerged from the initial  work on QCA of   Zeilinger~\cite{gz:qca},  Durr and Santha~\cite{dls:dpwfqca}, Watrous~\cite{a:watrous}, and  Meyer~\cite{bib:meyer2,dm:u1dqca}, the latter also bringing to fore an  important multi-particle subclass of QCA, the quantum lattice gas automaton (QLGA). These  papers establish conditions on QCA  existence and unitary evolution, and reveal their dynamics and   potential for simulating fundamental physics.

Axiomatic descriptions of QCA reflect the standard notions of quantum mechanics.  Axioms of  quantum mechanics~\cite{pd:pqm,vn:mgdq} are stated in a Hilbert space version and  a $C^*$ algebra  version.  In the Hilbert space version, the  state is a density operator, and the observables are bounded self-adjoint operators on the Hilbert space.  In the  $C^*$ algebra version, the state is a normalized linear functional on the  algebra.   Within the Hilbert space formulation, of abiding interest  are separable Hilbert spaces. They  have played a primal role in  quantum mechanics, and Gleason's crucial results~\cite{g:mcshs}  on quantum measurements   also rest on them. von Neumann and Murray~\cite{tv:hcwrop} develop rings of operators, {\it von Neumann algebras}, formalism of quantum mechanics.  Along with $C^*$ algebras from the ensuing work by Gelfand and Neumark~\cite{GelNeu43}, and Segal~\cite{segal1947}, these  algebras maintain an enduring significance in  algebraic quantum field theory (QFT) and infinite quantum systems.

 Schumacher and Werner~\cite{sw:rvqca} proposed an axiomatic model of a QCA in the operator formalism. A   QCA  resides on a lattice, the cells being  identified with the points of the lattice.    Each cell takes values in a finite dimensional $C^*$ algebra of operators,  the {\it cell algebra},  which is thus the endomorphisms of a finite dimensional Hilbert space.   In it, the cell algebras  are constituents of an  \textit{incomplete infinite  tensor product algebra}~\cite{vn:idp,js:its,ag:shrt,br:oaqsm1} (ITPA), a $^*$-aglebra of operators over the lattice. The system evolves at each time step through an  automorphism or this  $^*$-algebra. To be a valid model of physics, such an automorphim is to be  translation invariant, i.e., spatially homogeneous, acting identically on all cells. It must also  obey  {\it causality}, referring to a finite speed of information transfer, in  that  the image of a cell algebra under the automorphism is the identity operator on all but  a finite {\it neighborhood} of the cell.  This evolution of  the operator algebra is deemed to be in the Heisenberg picture of quantum mechanics.  It has led to  a classification of one-dimensional QCA by an {\it index theory} in the work of Gross, Nesme, Vogts, and Werner~\cite{gnvw:itodqwca}. We refer to this QCA model including its evolution as the  {\it Heisenberg  template} of QCA.

 Arrighi, Nesme, and Werner~\cite{anw:odqca} describe  an axiomatic  QCA in the Hilbert space formalism. Each cell  takes values in the {\it states} of a finite dimensional Hilbert space,  the {\it cell Hilbert space}.  A  typical basis element  of the QCA Hilbert space  consists of finitely many cells taking arbitrary values in the cell Hilbert space with a background of cells taking a designated identical {\it quiescent} value. The QCA Hilbert space is called  the {\it Hilbert space of finite, unbounded configurations} (UFC Hilbert space). The state of the QCA evolves at each time step by a unitary, translation-invariant and {\it causal}   operator. Causality, in this description, refers to a finite speed of information propagation at each step in the sense of evolution of the state: the state of each cell  after a step of evolution depends on that of the cells in its  finite neighborhood  prior to the step. This QCA, because of   unitary evolution, is in the Schr\"{o}dinger picture of quantum mechanics.  Among  elegant results related to this formalism is the work of Arrgihi, Nesme and Werner on localizabitlity of quantum cellular automata (QCA)~\cite{anw:ucil}. We refer to this QCA model including its evolution as the  {\it Schr\"{o}dinger  template} of QCA.

Templates of QCA serve as primary points of investigations concerning properties, structures  and classes of QCA that are potentially significant for quantum algorithms and simulations. Beyond classification of  QCA structure, much work has gone into modeling physics with QCA,  their multi-particle subclass the QLGA, and single-particle subclass, the quantum walks (QWs). Meyer~\cite{bib:meyer2} modeled Dirac equation in $1+1$ dimensions with QW, and Arrgihi et al~\cite{adma:deqwhctl}  in  $2+1$ dimensions. Boghosian and Taylor~\cite{bt:qlgmpsed} developed simulation models for many particle Schr\"{o}dinger equation based on QLGA.  QWs incorporating curved spacetime  appear in Arrighi et al~\cite{agf:qwcst}, and Di Molfetta et al~\cite{mbd:qwmdfcst}. Bisio et al~\cite{bat:qfqcad,bapt:thqca} develop models of  quantum field theory (QFT) as QCA. Arrgihi et al~\cite{am:qcgd} combine QCA with reversible causal graph dynamics as a framework for discrete time quantum gravity models. 
Meyer and Shakeel~\cite{bib:msqcawp}, seeking ways to simulate physics without particle description, combine QLGA to create QCA that no longer have a particle description at the  time scale  at which the dynamics are homogeneous.      Schlingemann et al~\cite{svw:oscqca} introduced Clifford QCA (CQCA) and investigated their structure, properties and generalizations.   G\"utschow et al~\cite{guwz:aegcqca} studied CQCA further for their entanglement generation properties. Crossing over into several areas, CQCA are  illustrative  of QCA intersecting fundamental physics, quantum computation and quantum resource theory. Interestingly, CQCA, described in the operator formalism, do not  readily have a  particle description in the sense of, for instance, QLGA, or  without particles  in the sense of the QCA of~\cite{bib:msqcawp}.  Interested reader can find in~\cite{pa:aoqca} a  comprehensive contemporary perspective and the   
significance of quantum cellular automata (QCA) as models of computation and simulation, their structural interpretations and analysis.

In their  paper on QFT simulation, Jordan et al~\cite{bib:JLP} calculate in polynomial time the  scattering amplitudes of particles in $\phi^4$  theory using a  quantum gate array architecture.  A question, raised in~\cite{bib:msqcawp}, is about the efficiency attainable in simulating  physics without a particle description. QCA are contenders among those simulation models, and  bringing the techniques available in operator formalism to bear on questions of this magnitude may be necessary.  Physical considerations, especially  implementation,  would  naturally demand that there be  a separable Hilbert space description with a  unitary state evolution. A result establishing that  a Schr\"{o}dinger template exists  equivalent to  any given  Heisenberg template of a QCA links the two approaches.  

It is a natural question to ask whether the templates are equivalent in the sense that each has a counterpart in the other with mirroring evolutionary dynamics. By this we mean that the measurement statistics are the same under evolution in either template, as in the Schr\"{o}dinger and Heisenberg pictures of quantum mechanical systems.  Given a Schr\"{o}dinger template over any lattice, we can always find a  Heisenberg template that has equivalent dynamics  by previous results mainly from~\cite{anw:odqca}.  Also, given a  Heisenberg  template over a {\it finite} sized lattice, we can always find a  Schr\"{o}dinger template that yields equivalent dynamics, by the Skolem-Noether Theorem for finite dimensional simple algebras, the relevant version of which is  cited as   Scholium~\ref{scholium} (Appendix~\ref{apprt}). There is no general theorem that directly gives this equivalence for countably infinite  dimensional vector spaces, therefore, it was not known if, given a Heisenberg template over an infinite lattice, an equivalent Schr\"{o}dinger  template exists. In this paper we  show that  there is an equivalent Schr\"{o}dinger template for every Heisenberg template. In fact, there is such a template for any UFC Hilbert space in whose bounded linear operators the  ITPA of the Heisenberg template is naturally embedded.
 
  This paper is organized as follows. Section~\ref{section:qcasec} describes the Schr\"{o}dinger and the Heisenberg templates of QCA. 
In  Section~\ref{section:intqca}  we first recall, without explicitly involving representation theory,  the known result that given a Schr\"{o}dinger template of a QCA, a dynamically equivalent Heisenberg  template can be constructed. This leads to the formal definition of equivalence   between the Schr\"{o}dinger and Heisenberg templates of  QCA in representation theoretic terms. Then we present the main result of this paper that given any Heisenberg template, an equivalent  Schr\"{o}dinger template exists for any UFC Hilbert space in whose bounded linear operators the ITPA of the Heisenberg template embeds.  Section~\ref{appl} examines the definition of a  QLGA in~\cite{sl:wqcqlg}, and its characterizing local condition, informed by the  results of this paper.
  Section~\ref{concl} is the conclusion.

\section{Quantum Cellular Automaton} \label{section:qcasec}
We recall the definition of Quantum Cellular Automaton (QCA) as it is defined in~\cite{anw:odqca} in the Hilbert space formulation, the  Schr\"{o}dinger template, and in the operator algebra formulation, the  Heisenberg template. 
\subsection*{The  Schr\"{o}dinger Template}
The  Schr\"{o}dinger template  in this paper is  a  slight generalization of the definition in~\cite{anw:odqca, sl:wqcqlg} in terms of spatial homogeneity, i.e.,  translation invariance.  Let us describe the Hilbert space of the QCA. 
Over each cell resides an identical finite dimensional complex Hilbert space, the {\it cell Hilbert space}, $W$.  We designate  a state $\ket{q} \in W$ as the \textit{quiescent} state,     and choose some orthonormal basis $\mathcal{B}$ of $W$ containing the quiescent state $\ket{q} \in \mathcal{B}$.  Each  basis element of the QCA Hilbert space is a tensor product, over the lattice, of basis elements from $\mathcal{B}$, in which  except for a finitely many cells, the rest are in the quiescent  state.  This basis is called the  \textit{set of finite, unbounded configurations},  denoted by ${\mathcal{C}}$, defined to be 
\begin{equation} \label{sofc} 
\mathcal{C} = \{ \bigotimes_{x \in \mathbb{Z}^n} \vert b_x \rangle :  \vert b_x \rangle \in \mathcal{B},  \text{ all but finite }\vert b_x \rangle   = \ket{q}  \} \:.
\end{equation}
Let the vector space spanned by $\mathcal{C}$ be $\mathbb{V}$,
\begin{equation} \label{Vfb}
\mathbb{V}=\text{Span}\big(\mathcal{C}\big) \:.
\end{equation} 

We can obtain $\mathbb{V}$ in other ways as well. Let an ascending chain of finite subsets be  $\{D_k \subset \mathbb{Z}^n : |D_k| < \infty\}_{k\in \mathbb{N}}$, such that  their  union is the entire lattice, i.e.,
\begin{equation} \label{Dkchain}
D_0 \subsetneq D_1 \subsetneq \ldots 
\end{equation}
such that $\mathbb{Z}^n = \cup_{k \in \mathbb{N}} D_k$.

Let
\begin{equation*}
W_{D_k} = \text{Span} (\{\bigotimes_{x \in D_k} \ket{b_x} \otimes  \bigotimes_{x \in \mathbb{Z}^n \setminus D_k} \ket{q} : \ket{b_x} \in \mathcal{B}\})\:.
\end{equation*}
 Under the natural   inclusion of $W_{D_k}$ into $W_{D_{k+1}}$,
\begin{equation*}
W_{D_k} \hookrightarrow W_{D_{k+1}} \:.
\end{equation*}
$W_{D_k}$ can be regarded as a subspace of $W_{D_{k+1}}$.
Then 
\begin{equation*} 
\mathbb{V} = \cup_{k} W_{D_k} \:,
\end{equation*} 
Another precise constrcution of  $\mathbb{V}$ is as  a direct limit in the manner  of Guichardet~\cite{ag:shrt}. The direct limit construction is in Appendix~\ref{backres}, and is  instructive for the rest of this paper.
The {\it Hilbert space of finite, unbounded configurations} (UFC Hilbert space), denoted by $\mathcal{H}$, is the completion of $\mathbb{V}$ under the inner product norm, i.e., under the  pre-Hilbert structure  on $\mathbb{V}$ induced  from the inner product on $W$.

The state of a QCA is   a  \textit{density operator} on  the UFC Hilbert space $\mathcal{H}$. A density operator is a positive trace class operator with trace $1$.  For the measurements that only concern  a finite  subset $D \subset \mathbb{Z}^n$ of cells,  the density operator can be {\it restricted} to $D$. This restriction is obtained by  a  \textit{partial trace} over cells (tensor factors) not in $D$. For that, let us write the  $\mathcal{H}$ as a tensor product of two Hilbert spaces, where completion in the inner product norm induced by those on the two tensor factors is assumed.
We can write,
\begin{equation} \label{VDdef}
\mathbb{V} = \big( \bigotimes_{x \in D} W \big)\otimes \mathbb{V}_{\bar{D}} \:,
\end{equation}
where $\mathbb{V}_{\bar{D}}$ is the vector space complement to the tensor factors in $D$.
Also,
\begin{equation*}
\mathcal{H} = \big( \bigotimes_{x \in D} W \big)\otimes \mathcal{H}_{\bar{D}} \:,
\end{equation*}
where $\mathcal{H}_{\bar{D}}$ is the Hilbert space completion of $\mathbb{V}_{\bar{D}}$. Let $\{ \vert v_i \rangle \}$ be some orthonormal basis  of $\mathcal{H}_{\bar{D}}$. The restricted density  operator is  calculated by first writing $\rho$ as
\begin{equation*}
\rho = \sum_{i,j} \rho_{i,j} \otimes \vert v_i \rangle \langle v_j \vert \:,
\end{equation*}
where $\rho_{i,j} \in  \begin{rm}{End}\end{rm}(\bigotimes_{x \in D} W)$.  Then $\rho$ restricted to $D$ is
\begin{equation*}
\rho |_{D} := \sum_{i} \rho_{i,i} \:.
\end{equation*}

Phase factors multiplying unitary evolution operators have no effect on measurements, so we are justified in the following definition.
\begin{defn} \label{classeq}
Let $R$ and  $R'$ be  operators on a UFC Hilbert space. Then $R$ and  $R'$ are {\it equivalent} if $R  = e^{i\theta} R'$ for some $\theta \in \mathbb{R}$. The equivalence class of $R$ is denoted  by $[R]$.
\end{defn}
We will refer to  $[R]$ as an operator on a UFC Hilbert space  when it is appropriate not to distinguish among the members of the class. When we say  $[R]$ is a unitary operator, we mean some  (class) representative $R$ of $[R]$ is a unitary operator. By the definition of $[R]$ then all representatives of $[R]$ are unitary.
In keeping with this equivalence, a  concept that we will generalize  from  the model  of QCA in~\cite{anw:odqca, sl:wqcqlg}  is that of spatial homogeneity, i.e.,  translation invariance,  in the definition of the Schr\"{o}dinger template.

The   {\it  neighborhood} of  cell $z \in \mathbb{Z}^n$ is a finite set,   denoted by $\mathcal{N}_z \subset \mathbb{Z}^n$. The state of the cell $z$ (restriction of the state to the cell) after a step of evolution depends on the state of the neighborhood $\mathcal{N}_z$ prior  the step, in the sense of  the Schr\"{o}dinger template definition below. For spatially homogeneous evolution, it is given by  translating a fixed finite subset  $\mathcal{N} \subset \mathbb{Z}^n$ to the cell $z$,
\begin{align} \label{nbhddef}
\mathcal{N}_z = z+ \mathcal{N} = \{z+y : y \in \mathcal{N}\} \:.
\end{align} 
 Then $\mathcal{N}$ is called the {\it neighborhood} for the QCA.
\begin{defn} [Schr\"{o}dinger template] \label{qcadef}
 Let   $W$ be a finite dimensional Hilbert space.  The  {\it Schr\"{o}dinger template} of a QCA is a pair ($\mathcal{H}$, $[R]$), where ${\mathcal{H}}$ is  a UFC Hilbert space constructed from $W$, and  $[R]$ is a  unitary operator on ${\mathcal{H}}$. $[R]$ is required to be
\begin{enumerate}
\item \label{transinvqca} \textit{Translation invariant}:  A translation operator $\tau_z$, for some $z \in \mathbb{Z}^n$, is a linear operator on  ${\mathcal{H}}$,  defined on a basis element $\vert c \rangle = \bigotimes_{x \in \mathbb{Z}^n} \vert b_x\rangle \in {\mathcal{C}}$ as:
\begin{equation*}
\tau_z:  \bigotimes_{x \in \mathbb{Z}^n} \vert b_x \rangle  \mapsto  \bigotimes_{x \in \mathbb{Z}^n} \vert b_{x+z} \rangle \:.
\end{equation*}
$[R]$ is \textit{translation invariant} if, for some  representative $R$ of $[R]$, and for all $z \in  \mathbb{Z}^n$,  $\tau_z   R   \tau^{-1}_z = e^{i \theta_z} R $, where $\theta_z$ depends on $z$.  
\item \label{causalqca} \textit{Causal relative to  some neighborhood $\mathcal{N}$}: 
$[R]$ is  \textit{causal} relative to the  neighborhood $\mathcal{N}$ if  for some  representative $R$ of $[R]$,  all  $z \in \mathbb{Z}^n$, and each pair $\rho, \rho'$  of   density operators on ${\mathcal{H}}$ satisfying
\begin{equation*}
\rho|_{\mathcal{N}_z} = \rho'|_{\mathcal{N}_z} \:,
\end{equation*}
$R \rho R^\dag, R\rho' R^\dag$  satisfy
\begin{equation*}
R \rho R^\dag |_z =  R \rho' R^\dag |_z  \:.
\end{equation*}
\end{enumerate}
 \end{defn}
The neighborhood $\mathcal{N}$ is typically chosen to be the smallest set that satisfies the causality condition given in  the   Schr\"{o}dinger template definition.  $[R]$ is called the  {\it global evolution operator} of the QCA. If the state of QCA is $\rho$ before a time step, then after the time step the state  evolves to  $R\rho R^\dag$, where $R$ is some representative of $[R]$.

\subsection*{The Heisenberg Template} The Heisenberg template of a QCA is defined on an algebra called the {\it incomplete infinite   tensor product algebra}. This algebra  is constructed from algebras that are {\it local} on finite subsets of the lattice $\mathbb{Z}^n$,  defined as follows.
\begin{defn} \label{localD}
Let  $W$ be a finite dimensional Hilbert space.  For a  finite subset  $D \in \mathbb{Z}^n$, the algebra of operators local upon $D$ is defined to be  $\mathcal{A}_D = \bigotimes_{x \in D}\begin{rm}{End}\end{rm}(W) \otimes \bigotimes_{y \in\mathbb{Z}^n \setminus D} \mathbb{I}$ \:. 
\end{defn}
\begin{remark} We denote identity operator by  $\mathbb{I}$ for every vector space, to economize  notation.
\end{remark}
\begin{remark} \label{csrem}{We write  $\mathcal{A}_z$ to mean  $\mathcal{A}_{\{z\}}$, the {\it cell algebra} of operators local upon one tensor factor $z$}.
\end{remark} 
Take an ascending chain of finite subsets as in~\eqref{Dkchain}.
 Under the obvious inclusion of $\mathcal{A}_{D_k}$ into $\mathcal{A}_{D_{k+1}}$ \:,
\begin{align*}
\mathcal{A}_{D_k} &\hookrightarrow \mathcal{A}_{D_{k+1}} \:,
\end{align*}
$\mathcal{A}_{D_k}$ can be regarded as a subalgebra of $\mathcal{A}_{D_{k+1}}$.
The  {\it incomplete infinite   tensor product algebra} (ITPA), which we  denote by $\mathcal{Z}$,  is
\begin{equation} \label{Zdef}
\mathcal{Z} = \cup_{k} \mathcal{A}_{D_k} \:.
\end{equation} 
Just as for the vector space of finite configurations~\eqref{Vaslim}, we can obtain $\mathcal{Z}$ as  direct limit of $\mathcal{A}_{D_k}$ (see Guichardet~\cite{ag:shrt}).


We note that  $\mathcal{Z}$ is  a $^*$-algebra since  each of its  constituent cell algebras,  $\mathcal{A}_z =  \begin{rm}{End}\end{rm}(W)$, is a  $^*$-algebra with the  adjoint map on operators as its  $^*$-involution,
\begin{align*}
  ^*: \begin{rm}{End}\end{rm}(W) &\rightarrow \begin{rm}{End}\end{rm}(W) \\
  a &\mapsto a^\dag \:.
\end{align*} 
This map extends to  a $^*$-involution on $\mathcal{Z}$.

Based on the constructs just introduced, we define the Heisenberg template of a QCA. 
 \begin{defn} [Heisenberg  template] \label{hqcadef}
Let  $W$ be a finite dimensional Hilbert space. The  {\it Heisenberg template} of a  QCA  is the  is  a pair ($\mathcal{Z}$, $\gamma$), where  $\mathcal{Z}$ is the ITPA~\eqref{Zdef}, and $\gamma$ is a   $^*$-automorphism     of $\mathcal{Z}$. $\gamma$ is required to be:
\begin{enumerate} 
\item \label{ithm1} \textit{translation invariant}: A translation operator $\mu_z$, for some $z \in \mathbb{Z}^n$, is a linear operator on $\mathcal{Z}$, defined on a basis  element $ b = \bigotimes_{x \in \mathbb{Z}^n}  b_x \in {\mathcal{Z}}$ as:
\begin{equation*}
\mu_z:  \bigotimes_{x \in \mathbb{Z}^n}  b_x   \mapsto  \bigotimes_{x \in \mathbb{Z}^n}  b_{x+z} \:. 
\end{equation*}
$\gamma$ is \textit{translation invariant} if  $\mu_z   \gamma   \mu^{-1}_z = \gamma $ for all $z \in  \mathbb{Z}^n$. 
\item \label{ithm2} \textit{Causal relative to some neighborhood $\mathcal{N}$}:  For every  element $A_z \in \mathcal{Z}$ local upon  $z$, $\gamma(A_z)$ is local upon  $\mathcal{N}_z$ (see Definition~\ref{localD}). 
\end{enumerate}
At each time step of evolution, $b \in \mathcal{Z}$ evolves to $\gamma(b)$.  
\end{defn}
Just as in the Schr\"{o}dinger template, Definition~\ref{qcadef}, the  neighborhood $\mathcal{N}$ is typically chosen to be the smallest subset of $\mathbb{Z}^n$ satisfying the causality condition in the above definition.

\section{Equivalence of the Schr\"{o}dinger and Heisenberg Templates of QCA} \label{section:intqca}
In this section, we address  the relation between  the  Schr\"{o}dinger and Heisenberg templates.  First, let us assume that we are given a Schr\"{o}dinger template $(\mathcal{H}, [R])$, with neighborhood $\mathcal{N}$ as in Definition~\ref{qcadef}, where $\mathcal{H}$ is constructed from a finite dimensional Hilbert space $W$. Then the ITPA $\mathcal{Z}$ constructed from $\text{End}(W)$ is embeddeded in the algebra of bounded linear operators $B(\mathcal{H})$. The following useful result  (Theorem $3.9$ from~\cite{sl:wqcqlg}) says that  conjugation of $\mathcal{Z}$ by $R$ is an automorphism of  $\mathcal{Z}$. 
\begin{thm}\label{thmsubalg}
Let  $R$ be a unitary and causal operator on $\mathcal{H}$ relative to some neighborhood $\mathcal{N}$.  Then for every $x\in \mathbb{Z}^n$,  
\begin{equation*}
\mathcal{A}_x \subset  \begin{rm}{span}\end{rm}(\prod_{y \in \mathcal{N}}R^\dag \mathcal{A}_{x-y} R) \:.
 \end{equation*}
In particular $ \mathcal{Z} = R^\dag\mathcal{Z}R$.

\end{thm}
Let the automorphism  $\gamma_{[R]}$ be,
 \begin{align} \label{gammaR}
\gamma_{[R]}: \mathcal{Z} &\longrightarrow  \mathcal{Z}\\
b  &\mapsto R^\dag b R \:,  \nonumber 
\end{align}
where $R$ is some representative of $[R]$.  $\gamma_{[R]}$  is a $^*$-automorphism as it commutes with the $^*$-involution of $\mathcal{Z}$ (adjoint map). 

We recall a theorem from~\cite{anw:odqca} deriving  the counterpart of causality  in the Heisenberg template from that in the Schr\"{o}dinger template.
Let the  {\it reflected neighborhood} of cell $z$ be $\mathcal{V}_z$, 
\begin{align*} 
\mathcal{V}_z = z- \mathcal{N} \:.
\end{align*} 
\begin{thm}[Structural Reversibility, Theorem $3$ in~\cite{anw:odqca}]  \label{strucreva}
  Let $R : \mathcal{H} \longrightarrow \mathcal{H}$ be a unitary  operator  and $\mathcal{N}$ a neighborhood.  Then the following are equivalent:
 \begin{enumerate}
\item \label{strv1} $R$ is causal relative to the  neighborhood $\mathcal{N}$. 
\item \label{strv2}  For every  operator $A_z$ local upon cell $z$, $R^\dag A_z R$ is local upon  $\mathcal{N}_z$.
\item \label{strv3} $R^\dag$ is causal relative to the  neighborhood $\mathcal{V}$. 
\item \label{strv4}  For every  operator $A_z$ local upon cell $z$, $R A_z R^\dag$ is local upon  ${\mathcal{V}}_z$.
\end{enumerate}   
\end{thm}
Part~\ref{strv2} of  this theorem implies that $\gamma_{[R]}$~\eqref{gammaR} is causal in the Heisenberg template sense, with the same neighborhood $\mathcal{N}$ as in the Schr\"{o}dinger template $(\mathcal{H}, [R])$. 
We can also show that $\gamma_{[R]}$   is  translation invariant in the Heisenberg template sense. This follows because $[R]$ is translation invariant, and because for all  $b \in \mathcal{Z}$, 
\begin{equation*} 
\mu_z(b) =  \tau_z b \tau_z^{-1},
 \end{equation*}
where $\mu_z$ is the translation operator in the Heisenberg template, and $\tau_z$ is the translation operator in the Schr\"{o}dinger template (Definition~\ref{qcadef}~\ref{transinvqca}). Consequently, for any representative $R$ of $[R]$,
\begin{align*}
\mu_z \gamma_{[R]} \mu_z^{-1} (b)  &=  \tau_z   R^\dag \tau_z^{-1} b \tau_z R \tau_z^{-1}  \\
&=   e^{-i\theta_z} R^\dag  b e^{i\theta_z} R \\
&=   R^\dag  b R = \gamma_{[R]}(b) \:.
\end{align*} 
So $\gamma_{[R]}$ is  translation invariant.

From the above discussion, it is clear that given a global evolution operator $[R]$ on $\mathcal{H}$, a $^*$-automorphism   $\gamma_{[R]}$ of  $\mathcal{Z}$  always exists.  This implies that  given a Schr\"{o}dinger  template, we can always obtain a Heisenberg  template (Definition~\ref{hqcadef}) with equivalent dynamics, i.e., yielding the same measurement statistics as in the Schr\"{o}dinger  template.

Given a Heisenberg template, the condition that determines if a dynamically equivalent  Schr\"{o}dinger template exists is best phrased in the  language of representations. 
The  obvious irreducible action of $\mathcal{Z}$   on $\mathbb{V}=\text{Span} (\mathcal{C}) $~\eqref{Vfb},
\begin{align} \label{piaction}
\mathcal{Z} \times \mathbb{V}  &\mapsto \mathbb{V}  \nonumber \\
\bigg(\bigotimes_{x\in \mathbb{Z}^n} A_x, \bigotimes_{x\in \mathbb{Z}^n} \ket{b_x} \bigg) &\rightarrow \bigotimes_{x\in \mathbb{Z}^n} A_x \big( \ket{b_x} \big) \:, 
\end{align}
extends to an irreducible representation on $\mathcal{H}$.  We denote this canonical representation ($\pi, {\mathcal{H}}$). This representation can be obtained as a direct limit, as shown in Guichardet~\cite{ag:shrt}.

Before we state the definition of equivalence, we need to be certain  that properties like translation invariance and causality would carry over from the Heisenberg template to the Schr\"{o}dinger template. 
\begin{prop} \label{SHthm}
Let $W$ be a finite dimensional Hilbert space, and ($\mathcal{Z}$, $\gamma$) be a   Heisenberg template, where $\mathcal{Z}$ is constructed from $\begin{rm}End\end{rm}(W)$. Suppose there exists a UFC Hilbert space $\mathcal{H}$ constructed over $W$, and a unitary operator $[R]$ on  $\mathcal{H}$ that  intertwines  the  representations  $(\pi\gamma, \mathcal{H})$ and $(\pi, \mathcal{H})$ of $\mathcal{Z}$, i.e., some representative $R$ of $[R]$ satisfies 
\begin{equation*} 
 R^\dag \pi(b) R  =  \pi  \left ( \gamma(b) \right)  \quad \forall b \in \mathcal{Z} \:.
 \end{equation*}
Then
\begin{enumerate}
\item $[R]$ is  unique.
\item $[R]$ is  translation invariant.
\item If $\gamma$ is causal relative to the neighborhood  $\mathcal{N}$, then  $[R]$ is  causal relative to the same neighborhood  $\mathcal{N}$.
\end{enumerate} 
 \end{prop}

\begin{proof}

\begin{enumerate}
\item $R$  is unique up to a phase factor  by Schur's Lemma (Lemma~\ref{srlem}, Appendix~\ref{apprt}) as the  representations $(\pi, \mathcal{H})$ and $(\pi\gamma, \mathcal{H})$ are irreducible; thus, $[R]$ is  unique.
\item Observe that for all $b \in \mathcal{Z}$ and all $z \in \mathbb{Z}^n$,
\begin{equation*} 
\pi(\mu_z(b)) =  \tau_z \pi(b) \tau_z^{-1} \:,
 \end{equation*}
which implies
\begin{align*} 
\pi(\mu_z \gamma \mu_z^{-1} (b)) &=  \tau_z \pi(\gamma \mu_z^{-1} (b)) \tau_z^{-1} \\
&=\tau_z   R^\dag \pi(\mu_z^{-1} (b)) R \tau_z^{-1} \\
&=\tau_z   R^\dag \tau_z^{-1} \pi(b) \tau_z R \tau_z^{-1} \:.
\end{align*}
Since   $\gamma$ is  translation invariant,
\begin{align*}
\pi(\gamma(b)) &=  \pi(\mu_z \gamma \mu_z^{-1} (b)) \\
&= \tau_z   R^\dag \tau_z^{-1} \pi(b) \tau_z R \tau_z^{-1} \:.
\end{align*} 
This implies that  $\tau_z R \tau_z^{-1}$ intertwines  the  representations  $(\pi, \mathcal{H})$ and  $(\pi\gamma, \mathcal{H})$. But $R$ is the unique intertwiner  up to a phase factor, so
\begin{equation*} 
\tau_z R \tau_z^{-1} =  e^{i\theta_{z}} R \:.
 \end{equation*}
for some $\theta_{z} \in \mathbb{R}$. Thus $[R]$ is  translation invariant. 
\item This follows  by the equivalence of  Theorem~\ref{strucreva} parts \ref{strv1} and~\ref{strv2}.
\end{enumerate}
\end{proof}
We are now justified in defining  the equivalence of templates as follows.
\begin{defn}
Let $W$ be a finite dimensional Hilbert space.  A Heisenberg template  ($\mathcal{Z}$, $\gamma$), where $\mathcal{Z}$ is constructed from $\text{End}(W)$, and a   Schr\"{o}dinger template    ($\mathcal{H}$,$[R]$), where $\mathcal{H}$ is  constructed from $W$,  are {\it equivalent} if  $[R]$ intertwines  the  representations  $(\pi\gamma, \mathcal{H})$ and $(\pi, \mathcal{H})$ 
of $\mathcal{Z}$, i.e., some  representative $R$   of $[R]$ satisfies
\begin{equation*} 
 R^\dag \pi(b) R  =  \pi  \left ( \gamma(b) \right)  \quad \forall b \in \mathcal{Z} \:.
 \end{equation*}
\end{defn}

The following result shows that a  Schr\"{o}dinger template exists  for  any  QCA given as a Heisenberg template  and for any UFC Hilbert space in which the  ITPA of the Heisenberg template is naturally embedded.
\begin{thm} \label{ESthm} 
Let $W$ be a finite dimensional Hilbert space. Given a Heisenberg template ($\mathcal{Z}$, $\gamma$),  where $\mathcal{Z}$ is constructed from $\begin{rm}End\end{rm}(W)$,  an equivalent   Schr\"{o}dinger template exists for any UFC Hilbert space $\mathcal{H}$ constructed from $W$.
\end{thm}
\begin{proof}
We consider the vector space $\mathbb{V}=\text{Span}(\mathcal{C})$~\eqref{Vfb}, and show that there is an automorphism of $\mathbb{V}$ that intertwines the $\mathcal{Z}$ (ITPA) representations  $(\pi\gamma, \mathbb{V})$ and $(\pi, \mathbb{V})$. Then we extend it to the UFC Hilbert space $\mathcal{H}$.

Given that $\gamma$ is an automorphism of $\mathcal{Z}$, and $\mathcal{A}_x \cap \mathcal{A}_y  = \mathbb{C}\mathbb{I}$ for $x\ne y$, implies that,
\begin{equation*} 
\gamma(\mathcal{A}_x) \cap \gamma(\mathcal{A}_y)  = \mathbb{C}\mathbb{I} \qquad \text{for } x\ne y.
\end{equation*}
Consider, a general  finite subset  $D \subset \mathbb{Z}^n$.  The image of $\mathcal{A}_D$ under $\gamma$ is $\gamma(\mathcal{A}_D)$.
Let the commutant of  $\gamma(\mathcal{A}_D)$ in  $\mathcal{Z}$ be  $\gamma(\mathcal{A}_D)^{C}$,
\begin{align*} 
\gamma(\mathcal{A}_D)^{C}&=  \{T \in \mathcal{Z} : \gamma(a) T = T  \gamma(a) \text{ for all } a \in \mathcal{A}_D\}  \nonumber \nonumber \\
&=  \cup_{k} \text{Span}\big(\Pi_{y \in D_k \setminus D} \gamma(\mathcal{A}_y)\big)
\end{align*}

$\gamma(\mathcal{A}_D)$ is local upon (i.e., non-identity on, as in Definition~\ref{localD}) a finite set of tensor factor indices by the causality  of $\gamma$, required in Definition~\ref{hqcadef}~(\ref{ithm1}). Let $K \subset \mathbb{Z}^n$ be the smallest set such that $\gamma(\mathcal{A}_D) \subset \mathcal{A}_K =  \bigotimes_{y \in K} \mathcal{A}_y$. By~\eqref{gammaid} again we are justified in writing
  \begin{equation*} \label{gammaid}
\mathcal{Z} = \text{Span}\big(\gamma(\mathcal{A}_D)  \gamma(\mathcal{A}_D)^{C}\big).
\end{equation*}
$\mathcal{A}_K \subset \mathcal{Z}$, so
  \begin{align} \label{Akdec}
 \mathcal{A}_K &=  \text{Span}\big(\gamma(\mathcal{A}_D)  \gamma(\mathcal{A}_D)^{C}\big) \cap  \mathcal{A}_K  \nonumber \\
  &= \text{Span}\bigg(\gamma(\mathcal{A}_D) \big(\gamma(\mathcal{A}_D)^{C} \cap  \mathcal{A}_K\big)  \bigg) \:,
\end{align}
where the second equality follows since  $\gamma(\mathcal{A}_D) \subset \mathcal{A}_K$.  
  It allows us to consider  $\gamma(\mathcal{A}_D)$ as a subalgebra  of $\mathcal{A}_K$ instead of  $\mathcal{Z}$. We  call the  restriction of the  representation $(\pi \gamma, \mathbb{V})$ to a subalgebra  $\mathcal{A} \subset \mathcal{Z}$, the action of $\mathcal{A}$ via $\pi \gamma$ on $\mathbb{V}$. The action of  algebra $\mathcal{A}_D=\bigotimes_{x \in D} \mathcal{A}_x$  via $\pi \gamma$ on $\mathbb{V}$, because of this, can be restricted to    $\bigotimes_{x \in K} W$, and extended trivially to the rest of  $\mathbb{V}$. $\gamma(\mathcal{A}_D) \cong \bigotimes_{x \in D} \text{End} (W) = \begin{rm}{End}\end{rm}\big( \bigotimes_{x \in D} W \big)$   is a simple, self-adjoint algebra.
Let us  denote the canonical representation of $\text{End} (W)$ on $W$ by $(\chi, W)$.  Up to equivalence, the  finite dimensional irreducible module of $\mathcal{A}_D = \bigotimes_{x \in D} \text{End} (W)$ occurring in $\bigotimes_{x \in K} W$  is the {\it outer tensor product}~\eqref{outtenrep} representation  $(\widehat{\bigotimes}_{x \in D} \chi, \bigotimes_{x \in D} W)$. 
By the Double-Commutant Theorem, Theorem~\ref{dblctthm}, there is an  $\mathcal{A}_D$ module isomorphism, 
\begin{equation*} 
S_K : U  \otimes \big( \bigotimes_{x \in D} W \big)  \longrightarrow  \bigotimes_{x \in K} W \:,
\end{equation*}
where $U = \text {Hom}_{\mathcal{A}_D}(\widehat{\bigotimes}_{x \in D} \chi, \bigotimes_{x \in K} W)$ is the multiplicity space for  $\mathcal{A}_D$ action via $\pi\gamma$ on $\bigotimes_{x \in K} W$. Also, under this isomorphism,
\begin{equation*} 
\mathcal{A}_D \cong  \mathbb{I}  \otimes \begin{rm}{End}\end{rm}\big( \bigotimes_{x \in D} W \big) \:.
\end{equation*}
This implies that the irreducible representation of $\mathcal{A}_D$ in its action via $\pi \gamma$ on $\mathbb{V}$ is  the {\it outer tensor product}~\eqref{outtenrep} representation $(\widehat{\bigotimes}_{x \in D} \chi, \bigotimes_{x \in D} W)$, which is locally irreducible.
Therefore, by the primary decomposition~\eqref{pdec}, there is an $\mathcal{A}_D$ module isomorphism, 
\begin{align} \label{constRD}
S_D: E_D \otimes \big(\bigotimes_{x \in D} W \big) &\cong \mathbb{V} \nonumber \\
\epsilon_D \otimes \bigotimes_{x \in D} w_x &\mapsto \epsilon_{D_n}\big( \bigotimes_{x \in D} w_x\big) \:,
\end{align}
where $E_D = \text {Hom}_{\mathcal{A}_D}(\widehat{\bigotimes}_{x \in D} \chi, \mathbb{V})$  is the multiplicity space for $\mathcal{A}_D$ action via $\pi\gamma$. Denote by  $\mathcal{Z}^{D}$ the commutant of  $\mathcal{A}_D$ in  $\mathcal{Z}$  for action via $\pi \gamma$ on $\mathbb{V}$, 
\begin{align} \label{ZdefgammaD}
\mathcal{Z}^{D}&=  \{T \in \mathcal{Z} : \pi \gamma(a) \pi \gamma(T)=\pi \gamma(T) \pi \gamma(a) \text{ for all } a \in \mathcal{A}_D\}  \nonumber \\
&=  \{T \in \mathcal{Z} : aT=Ta \text{ for all } a \in \mathcal{A}_D\} \nonumber \\ 
&=  \cup_{k} \bigotimes_{y \in D_k \setminus D} \mathcal{A}_y \:.
\end{align}
$T\in \mathcal{Z}^{D}$~\eqref{ZdefgammaD} acts as $\gamma(T)$ on $E_D$ by left multiplication (see the discussion in Appendix~\ref{apprt} after Proposition~\ref{proppdec}). Let us denote by  $l$ the left multiplication by an element of $\gamma(\mathcal{Z}^{D})$ on $E_D$. Since $\gamma$ is an automorphism and $\mathbb{V}$ is an irreducible $\mathcal{Z}$  module, this implies $E_D \otimes (\bigotimes_{x \in D} W )$ is an irreducible  $\mathcal{Z}^{D} \otimes \mathcal{A}_D$ module  under the outer tensor product action  $(l \circ \gamma) \widehat{\otimes} \big(\widehat{\bigotimes}_{x \in D} \chi\big) $. Thus,  $E_D$ is an irreducible $\mathcal{Z}^{D}$ module under $l \circ \gamma$. We denote by $V_{D}$ the vector space on the left side of  equation~\eqref{constRD},
\begin{equation*} 
V_{D} =  E_{D}  \otimes \big(\bigotimes_{x \in D} W \big) \:.
\end{equation*}

Let us consider the canonical representation of $\mathcal{A}_D$ on $\mathbb{V}$ via $\pi$, i.e., restriction of $\pi$ to $\mathcal{A}_D$.  Then we have an $\mathcal{A}_D$ module isomorphism
\begin{align*} 
L_D: U_D \otimes \big(\bigotimes_{x \in D} W \big) &\cong \mathbb{V} \nonumber \\
u_D \otimes \bigotimes_{x \in D} w_x &\mapsto u_D\big( \bigotimes_{x \in D} w_x\big) \:,
\end{align*}
where $U_D = \text {Hom}_{\mathcal{A}_D}(\widehat{\bigotimes}_{x \in D} \chi, \mathbb{V})$ is the multiplicity space for $\mathcal{A}_D$ action via $\pi$. We can identify the multiplicity space $U_D$ in this case with $\mathbb{V}_{\bar{D}}$~\eqref{VDdef}, the vector space complement in $\mathbb{V}$ to the tensor factors in $D$. We are implicitly identifying the indices of  the modules $\bigotimes_{x \in D} W$ with those in $\mathbb{V}$ for the  $\mathcal{A}_D$ actions  via $\pi \gamma$ and $\pi$.

Choose  any map $F \in   \text {Hom}(E_D,U_D)$. Then $L_D (\mathbb{I}\otimes F) S_D^{-1} \in \text {End}_{\mathcal{A}_D}(\mathbb{V})$, 
i.e., is an $\mathcal{A}_D$ module homomorphism of $V$ intertwining $\mathcal{A}_D$ actions via $\pi\gamma$ and $\pi$. 
 Let us denote by $\mathcal{R}_D$ the vector space of $\mathcal{A}_D$ module homomorphisms of $V$ intertwining $\mathcal{A}_D$ actions via $\pi\gamma$ and $\pi$,
\begin{equation} \label{RDDEf}
\mathcal{R}_D = \text {End}_{\mathcal{A}_D}(\mathbb{V}) \:.
\end{equation}
We have now set up the objects we need for a  finite subset $D \subset \mathbb{Z}^n$.

Next, we specialize to the sets $\{D_k\}$ in  the ascending chain~\eqref{Dkchain} and  relate the spaces and morphisms associated with them. Let $D_m \subset D_n \in \{D_k\}$. We construct a  $\mathcal{Z}$ module isomorphism between the spaces $V_{D_m}$ and $V_{D_n}$, 
\begin{align} \label{VDmn}
V_{D_m} &=    E_{D_m}  \otimes \bigotimes_{x \in D_m} W,    \nonumber \\
V_{D_n} &=   E_{D_n}  \otimes  \bigotimes_{x \in D_n} W  \:.
\end{align}
There is a canonical map,
\begin{equation*} 
J_{n,m}:    E_{D_n} \otimes   \bigotimes_{x \in D_n \setminus D_m} W  \rightarrow E_{D_m} \:,     
\end{equation*}
that we need to  describe first.  Choose an orthonormal   basis $\mathcal{B} \subset W$  containing the quiescent state $\ket{q} \in W$, as in the definition of the set of finite, unbounded  configurations~\eqref{sofc}. 
 Observe that an element  $\epsilon_{D_n} \in  E_{D_n}$ is a linear map 
\begin{equation*}
\epsilon_{D_n}:  \bigotimes_{x \in D_n} W \rightarrow \mathbb{V} \:.
 \end{equation*}
We can write  $\bigotimes_{x \in D_n} W$ in the form
\begin{equation*}
 \bigotimes_{x \in D_n} W =   \bigotimes_{x \in D_m} W \otimes \bigotimes_{y \in D_n \setminus D_m} W \:,
\end{equation*}
expressing a basis element as 
\begin{equation*}
 \bigotimes_{x \in D_n} b_x =  \bigotimes_{x \in D_m} b_x \otimes \bigotimes_{y \in D_n \setminus D_m} b_y \:,
\end{equation*}
where $b_x \in \mathcal{B}$. 
Then $\epsilon_{D_n}$  maps such an element to $\mathbb{V}$,
\begin{align*}
\epsilon_{D_n}:  \bigotimes_{x \in D_n} W &\rightarrow \mathbb{V} \\ 
 \bigotimes_{x \in D_m} b_x \otimes \bigotimes_{y \in D_n \setminus D_m} b_y &\mapsto  \epsilon_{D_n}\big(\bigotimes_{x \in D_m} b_x \otimes \bigotimes_{y \in D_n \setminus D_m} b_y\big) \:.
\end{align*}
We obtain $J_{n,m}$ by evaluating elements of $E_{D_n}$ at the components corresponding to the indices in $D_n \setminus D_m$,
\begin{align} \label{Jnm}
  J_{n,m}:   E_{D_n} \otimes   \bigotimes_{y \in D_n \setminus D_m} W  &\rightarrow E_{D_m}  \nonumber \\ 
 \epsilon_{D_n}  \otimes  \bigotimes_{y \in D_n \setminus D_m} b_y &\mapsto  \epsilon_{D_m}\big(\bigotimes_{x \in D_m} b_x\big)=\epsilon_{D_n}\big(\bigotimes_{x \in D_m} b_x \otimes \bigotimes_{y \in D_n \setminus D_m} b_y\big) \:,
\end{align}
and extending by linearity. 
We extend $J_{n,m}$ to $V_{D_n}$ by  identity on the indices in $D_m$, and denote this map by ${\hat J}_{n,m}$,
\begin{equation*} 
{\hat J}_{n,m}: E_{D_n} \otimes   \bigotimes_{x \in D_n} W  \rightarrow E_{D_m}  \otimes   \bigotimes_{x \in D_m} W  \:.
\end{equation*}
${\hat J}_{n,m}$ intertwines the irreducible actions of $\mathcal{Z}$ on $V_{D_n}$ and $V_{D_m}$~\eqref{VDmn}. By Schur's Lemma (Lemma~\ref{srlem}), it is an  isomorphism. Note that this further implies that $J_{n,m}$~\eqref{Jnm} is an isomorphism. 
For $D_l \subset D_m \subset D_n$, we have that
\begin{equation} \label{Jrel}
{\hat J}_{n,l}=  {\hat J}_{m,l} \circ  {\hat J}_{n,m} \:. 
\end{equation}
Also, ${\hat J}_{k,k} = \mathbb{I}$ for all $k$. Let
\begin{equation} \label{Trel}
{\hat T}_{m,n} =  {\hat J}^{-1}_{n,m}: V_{D_m} \rightarrow  V_{D_n} \:.
\end{equation}
$\{V_{D_m}, {\hat T}_{m,n}\}$  form  a direct system (see Appendix~\ref{dirlimback}) of isomorphic $\mathcal{Z}$ modules, i.e., for $D_l \subset D_m \subset D_n$,  ${\hat T}_{l,n}={\hat T}_{m,n} \circ {\hat T}_{l,m}$, which follows from~\eqref{Jrel}.  The direct limit of this system exists by Theorem~\ref{dirlim}.  Let $\hat{\mathbb{V}}$ be this direct limit,
\begin{equation*} 
\hat{\mathbb{V}}  = \underrightarrow{\lim} V_{D_m}.
\end{equation*}
Observe that $\mathbb{V}$ also satisfies the universal property of a direct limit  for the system $(V_{D_m}, {\hat T}_{m,n})$. To see this, without loss of generality, include the empty set in the ascending chain of sets~\eqref{Dkchain}, $D_0=\emptyset \in \{D_k\}$. Then $V_{\emptyset}= \mathbb{V}$.  Let  $f^k={\hat T}^{-1}_{0,k}={\hat J}_{k,0}$ for  $D_k$. Then $(\mathbb{V}, \{f^k\})$ is a target, since, $f^m = f^n \circ {\hat T}_{m,n}$ for $D_m \subset D_n$. 

Let $(\mathbb{W},\{g^k\})$ be another target. Let $g=g^{0}:   V_{\emptyset}= \mathbb{V} \rightarrow \mathbb{W}$. Then $g^k=g \circ f^k$. The map $g$ is unique, i.e., if there is another map $h: V_{\emptyset}= \mathbb{V} \rightarrow \mathbb{W}$ satisfying $g^k=h \circ f^k$, then $h=g$, since $f^k$ are isomorphisms. This shows that $(\mathbb{V}, \{f^k\})$ satisfies the universal property. Thus $\mathbb{V}$ is isomorphic to $\hat{\mathbb{V}}$, which follows by the  uniqueness up to unique isomorphism of the direct limit. Let us denote this isomorphism by $\hat{S}$,
\begin{equation*} 
\hat{S}: \hat{\mathbb{V}}  \rightarrow \mathbb{V},
\end{equation*}
Let us denote by ($\hat{\pi}, \hat{\mathbb{V}}$)  the representation of  $\mathcal{Z}$ on  $\hat{\mathbb{V}}$.
Note that  the $\mathcal{Z}$  isomorphisms ${\hat T}_{m,n}$~\eqref{Trel} are defined with respect  to the canonical   isomorphism~\eqref{Slambda}, implying  that the  intertwining of actions is independent of the class representative in the  direct system.  $\hat{\mathbb{V}}$ is thus the  natural representation of $\mathcal{Z}$ on $\hat{\mathbb{V}}$ compatible with the tensor product structures of $\mathcal{Z}$ and $\hat{\mathbb{V}}$, that of $\hat{\mathbb{V}}$  coming from ${\hat T}_{m,n}$~\eqref{Trel}. $\hat{S}$ intertwines the irreducible $\mathcal{Z}$ representations   ($\hat{\pi}, \widehat{\mathbb{V}}$) and  ($\pi \gamma, \mathbb{V}$).

Next, we look for the $\mathcal{Z}$ module isomorphism of $\mathbb{V}$ intertwining the $\mathcal{Z}$ representations ($\pi \gamma, \mathbb{V}$) and  ($\pi, \mathbb{V}$). Under the isomorphisms in~\eqref{Jnm}, we have that 
\begin{equation*}
\text{Hom}(E_{D_m},U_{D_m}) \cong  \text{End}(\bigotimes_{x \in D_n \setminus D_m} W) \otimes \text{Hom}(E_{D_n},U_{D_n}) \:.
\end{equation*}

The trace, $\langle A,B \rangle= \text{tr} (AB)$, for $A,B \in \text{End}(\bigotimes_{x \in D_n \setminus D_m} W)$, is a non-degenerate bilinear form. Let $P_{m,n}$ be the linear operator
\begin{align*}
P_{m,n}: \text{End}(\bigotimes_{x \in D_n \setminus D_m} W) \otimes \text{Hom}(E_{D_n},U_{D_n}) &\rightarrow \text{Hom}(E_{D_n},U_{D_n})  \nonumber \\
\alpha \otimes \beta &\mapsto  \text{tr}(\alpha) \beta \:.
\end{align*}
$P_{m,n}$ is a surjective homomorphism. By the discussion before~\eqref{RDDEf}, under the isomorphisms in~\eqref{Jnm}, $P_{m,n}$ induces a surjective homomorphism of vector spaces,
\begin{equation*}
M_{m,n}: \mathcal{R}_{D_m} = \text {End}_{\mathcal{A}_{D_m}}(\mathbb{V}) \rightarrow  \mathcal{R}_{D_n} =  \text {End}_{\mathcal{A}_{D_n}}(\mathbb{V})\:.
\end{equation*}
For $\nu,\eta \in \text{End}(W)$, $\text{tr}(\nu\otimes \eta) = \text{tr}(\nu) \text{tr}(\eta)$. This implies,  for  $D_l \subset D_m \subset D_n$,  $M_{l,n}=M_{m,n} \circ M_{l,m}$. Thus, 
$\{\mathcal{R}_{D_m}, M_{m,n}\}$  form  a direct system. Let its direct limit, which exists by Theorem~\ref{dirlim}, be
\begin{equation*} 
\hat{\mathcal{R}}  = \underrightarrow{\lim} \mathcal{R}_{D_m}\:.
\end{equation*}
\begin{remark} The vector space $\mathbb{V}$ is a direct limit~\eqref{Vaslim}, the ITPA $\mathcal{Z}$  and the canonical representation $\pi$ can also be constructed as direct limits (Guichardet~\cite{ag:shrt}). The direct limits in this proof should be viewed in that context.
\end{remark}
Choose $\hat{R} \in \hat{\mathcal{R}}$,
\begin{equation*} 
\hat{R}: \mathbb{V}  \rightarrow \mathbb{V} \: .
\end{equation*}
Then $\hat{R}$ intertwines the $\mathcal{Z}$ representations  ($\pi \gamma, \mathbb{V}$) and ($\pi, \mathbb{V}$).  By Schur's Lemma (Lemma~\ref{srlem}), $\hat{R}$ is the  unique $\mathcal{Z}$ module isomorphism up to a phase factor. Since $\gamma$ is a $*$-automorphism,  $\hat{R}$ is  unitary on $\mathbb{V}$. Further,  by the Bounded Linear Transformation (BLT) Theorem (Reed and Simon~\cite{rs:fa1}  Theorem 1.7, pg. 9), $\hat{R}$ can be uniquely and unitarily extended to  $\mathcal{H}$. 
We denote  this  extension to $\mathcal{H}$ by $R$,
\begin{equation} \label{RdefH}
R: \mathcal{H}  \rightarrow \mathcal{H},
\end{equation}

By Proposition~\ref{SHthm},  the Schr\"{o}dinger template  given by ($\mathcal{H}$,$[R]$) is equivalent to the given Heisenberg template ($\mathcal{Z}, \gamma$).
\end{proof}

\section{Template equivalence and the quantum lattice gas automaton} \label{appl}
 A QLGA is a model of  multiple particles propagating on a lattice  and scattering at the lattice points (called {\it sites} in this context) through interactions with each other at the lattice points or through self-interactions.  In~\cite{sl:wqcqlg}, the authors  studied the conditions   that classify a QCA as a QLGA. The definition of  a  QLGA   given in~\cite{sl:wqcqlg} is a Schr\"{o}dinger template  based on a particular UFC Hilbert space.  The QLGA  classifying condition, however, involves  propagating cell algebras interacting  at a cell, belonging to the Heisenberg template.  Just as in the  main theorem of~\cite{sl:wqcqlg},  Theorem $\rm{III}.16$ in that paper,  we start from the QLGA calssifying condition of a QCA, and construct the same Schr\"{o}dinger template as in~\cite{sl:wqcqlg}, albeit in a mode that reflects the  approach of this paper. Through this construction, we will see restrictions being  imposed beyond those stipulated by the classifying condition to reach the explicit form required by the QLGA definition of~\cite{sl:wqcqlg}. We then scrutinize these restrictions, now apprised of  the more general  Schr\"{o}dinger  template implied by  Theorem~\ref{ESthm}.

Let us recall the definition of a QLGA from~\cite{sl:wqcqlg}. 
\begin{defn} \label{qlgadef}
A QLGA is defined on a lattice $\mathbb{Z}^n$,  and for a  neighborhood $\mathcal{N}$, as follows:
\begin{enumerate}
\item{The cell Hilbert space is  $ W = \bigotimes_{z\in\mathcal{N}} V_z$, for some finite dimensional vector spaces $\{V_{z}\}_{z \in \mathcal{N}}$. 

\item The \textit{quiescent state} $\vert q \rangle$, which  is a simple product,
\begin{equation*} 
\vert q \rangle = \bigotimes_{z \in \mathcal{N}} \vert q_{z} \rangle{\rm,~where}~\vert q_{z} \rangle  \in V_z
\end{equation*}
}

\item A UFC Hilbert space  $\mathcal{H}$ defined  in terms of $ W = \bigotimes_{z\in\mathcal{N}} V_z$ such that the quiescent states in all cells are identical  $\ket{q_x} = \ket{q}$ for all $x \in \mathbb{Z}^n$.
\item A {\it propagation operator} relative to the neighborhood $\mathcal{N}$,    $\sigma: \mathcal{H} \mapsto \mathcal{H}$, defined on the basis of $\mathcal{H}$ to be
\begin{equation}  \label{sigdefn}
 \left . \begin{array} {cccc}
 \sigma : 
  \bigotimes_{x \in \mathbb{Z}^n} \bigotimes_{z \in \mathcal{N}} \vert \textbf{k}_x(z) \rangle  \mapsto  \bigotimes_{x \in \mathbb{Z}^n} \bigotimes_{z \in \mathcal{N}} \vert \textbf{k}_{x+z}(z) \rangle
    \end{array}
  \right .
\end{equation} 
\item A local scattering operator $F$, which is a unitary operator  on the site Hilbert space $ W = \bigotimes_{z\in\mathcal{N}} V_z$, such that $F$ fixes  $\vert q \rangle$  (an eigenvector with eigenvalue one): $F \vert q \rangle =  \vert q \rangle$.  The global  \textit{scattering operator} $ \tilde F: \mathcal{H} \mapsto \mathcal{H}$, is the application of $F$ at every cell,  defined as
\begin{equation} \label{Fdef}
\tilde F =  \bigotimes_{x \in \mathbb{Z}^n} F 
\end{equation}
\item A global evolution operator $ R$ consisting of propagation $\sigma$ followed by the scattering $\tilde F$:
\begin{equation*}
 R   =    {\tilde F} \sigma 
\end{equation*}
\end{enumerate}
 A state of the QLGA is a vector in  the UFC Hilbert space.
\end{defn}

First, let us recall the terminology of~\cite{sl:wqcqlg}, but generalized to an arbitrary Heisenberg template automorphism $\gamma$.  The patch of propagated image $\gamma(\mathcal{A}_{z})$ on $\mathcal{A}_{x}$, is
\begin{equation*} 
\mathcal{D}_{z,x} = \gamma(\mathcal{A}_{z})   \cap \mathcal{A}_{x}.
\end{equation*}
We restate a version of  Theorem $\rm{III}.10$ of~\cite{sl:wqcqlg}, as  we will  refer to it. The  proof of this verison would be very similar to that of Theorem $\rm{III}.10$ in~\cite{sl:wqcqlg}, except  in the current context we have  $\gamma$, whereas in~\cite{sl:wqcqlg} the proof was in the context of $\gamma_{[R]}$~\eqref{gammaR}.
\begin{thm} [Theorem $\rm{III}.10$ in~\cite{sl:wqcqlg}] \label{thmS}
Let ($\mathcal{Z}$, $\gamma$) be the Heisenberg template of a QCA   with   neighborhood $\mathcal{N}$. Then 
$\mathcal{A}_x = \begin{rm}{span }\end{rm}(\prod_{y \in  \mathcal{N}} \mathcal{D}_{x-y,x})$
if and only if there exists an isomorphism of vector spaces:
\begin{equation} \label{Sisom}
S :  W \longrightarrow \bigotimes_{z \in \mathcal{N}} V_{z}
\end{equation}
for some vector spaces $\{V_{z}\}_{z \in \mathcal{N}}$.  Under the isomorphism $S$, for each $y   \in \mathcal{N}$:
\begin{equation*} 
\mathcal{D}_{x-y,x} \cong  \begin{rm}{End}\end{rm}(V_y) \otimes \bigotimes_{z \in \mathcal{N}, z \neq y} \mathbb{I}_{V_z}
\end{equation*}
\end{thm}

The condition $\mathcal{A}_x = \begin{rm}{span }\end{rm}(\prod_{y \in  \mathcal{N}} \mathcal{D}_{x-y,x})$  is the local characterization of the  Schr\"{o}dinger template of a QCA as a QLGA in Theorem $\rm{III}.16$ in~\cite{sl:wqcqlg}, although, it  is manifestly a condition on $\gamma$, and hence on the Heisenberg template. 
Assuming  this condition, we work toward a Schr\"{o}dinger template in which the global evolution operator is explicit.  First, we  replace $W$ with $S(W)$, where $S$ is the isomorphism in~\eqref{Sisom} and  can be taken to be an isometric isomorphism under an appropriate inner product choice on $S(W)$. Under this substitution, and abusing notation, we can assume that $W$ has the tensor product structure implied in~\eqref{Sisom}. 
Take the quiescent state $\vert q \rangle \in W  = \bigotimes_{z \in \mathcal{N}} V_z$ to be \textit{any} product vector in $W$,
 \begin{equation*}  
\vert{q} \rangle   = \bigotimes_{z \in \mathcal{N}} \vert{q}_z \rangle{\rm,~where}~\vert{q}_z \rangle  \in V_z. 
\end{equation*}
Using $W$ and  $\ket{q}$, we  construct the UFC Hilbert space $\mathcal{H}$ as described in Section~\ref{section:qcasec}.
Next, define a propagation operator $\sigma$ as in~\eqref{sigdefn}.
\begin{equation*} 
 \left . \begin{array} {cccc}
 \sigma : 
  \bigotimes_{x \in \mathbb{Z}^n} \bigotimes_{z \in \mathcal{N}} \vert \textbf{k}_x(z) \rangle  \mapsto  \bigotimes_{x \in \mathbb{Z}^n} \bigotimes_{z \in \mathcal{N}} \vert \textbf{k}_{x+z}(z) \rangle
    \end{array}
  \right .
\end{equation*}  
Then  $\sigma \pi \gamma ({\mathcal{A}}_x)  \sigma^{-1} = \pi ({\mathcal{A}}_x) $, i.e.,  $ \sigma \pi \gamma ({\mathcal{A}}_x) \sigma^{-1}$ is an automorphic image of  $\pi ({\mathcal{A}}_x)$. But  $\pi ({\mathcal{A}}_x) = \begin{rm}{End} \end{rm} (W)$. Therefore, there is a unitary map $F$ on $W$ (Schur's Lemma applied to ${\mathcal{A}}_x$ action on $W$), such that \begin{equation} \label{Fdefn}
\sigma \pi \gamma ({\mathcal{A}}_x) \sigma^{-1}  = F^{-1} \pi ({\mathcal{A}}_x) F
\end{equation} 
for each $x \in \mathbb{Z}^n$. By translation invariance of $\gamma$, the same $F$ works for every cell.  Note that $F$ only depends on $\gamma$ and not on the choice of $\ket{q}$. We say that $F$ is \textit{associated} with $\gamma$.

We state and prove the direction of the  main theorem in~\cite{sl:wqcqlg},  Theorem $\rm{III}.16$,  in which the local condition $\mathcal{A}_x = \begin{rm}{span }\end{rm}(\prod_{y \in  \mathcal{N}} \mathcal{D}_{x-y,x})$ on the Heisenberg template of a QCA implies that its  Schr\"{o}dinger template is a QLGA (upto a global isomorphism). 
An eigenvector  of an operator on a finite dimensional tensor product vector space  will be called a  \textit{product eigenvector} if it is a simple product in the vector space.
\begin{thm}[See Theorem $\rm{III}.16$ in~\cite{sl:wqcqlg}] \label{thmStructure}
Let ($\mathcal{Z}$, $\gamma$) be the Heisenberg template of a QCA   with   neighborhood $\mathcal{N}$,  satisfying: 
 $\mathcal{A}_x = \begin{rm}{span }\end{rm}(\prod_{y \in  \mathcal{N}} \mathcal{D}_{x-y,x})$. Then it is a QLGA in the sense of Definition~\ref{qlgadef} if and only if  $F$~\eqref{Fdefn} associated  with $\gamma$ has a product eigenvector in $W$. 
\end{thm}
\begin{proof}
Suppose that  $F$ has a product eigenvector  $\ket{w} \in W$. Since $F$ is unitary,  $F  \ket{w} =  e^{i\theta}\ket{w}$ for some  $\theta \in \mathbb{R}$. We can replace $F$ with $e^{-i\theta}F$, which fixes $\ket{w}$, and is still associated with $\gamma$ in~\eqref{Fdefn}. Without loss of generality then, let us assume that $F$ fixes  $\ket{w}$. Let  $\vert q \rangle = \ket{w}$. We can now construct a cell-wise automorphism  $\tilde F$ on $\mathcal{H}$ as in~\eqref{Fdef},
\begin{equation*}
\tilde F =  \bigotimes_{x \in \mathbb{Z}^n} F,
\end{equation*}
 As ${\mathcal{A}}_x$, $x \in \mathbb{Z}^n$,  generate $\mathcal{Z}$,  
\begin{equation*} 
\sigma \pi \gamma (b) \sigma^{-1} =  {\tilde F}^{-1}  \pi (b) \tilde F
\end{equation*}
for all $b \in \mathcal{Z}$.  
Rewriting the above relation, we obtain
\begin{equation*}
 \pi \gamma (b) =  \sigma^{-1} {\tilde F}^{-1}  \pi (b) \tilde F \sigma.
\end{equation*}
Thus
\begin{equation*}
R =     {\tilde F} \sigma      
\end{equation*}
intertwines $(\pi, \mathcal{H})$ and  $(\pi\gamma, \mathcal{H})$. By Proposition~\ref{SHthm}, $[R]$ is the unique such global evolution operator. Indeed it is a QLGA, as  it is composed  of the propagation operator $\sigma$ followed by the  scattering operator  $\tilde F$. 

Conversely, if $F$ does not have a product eigenvector in $W$,  that precludes defining  a  quiescent state $\ket{q} \in W$ from which to  construct a UFC Hilbert space $\mathcal{H}$: both $\sigma$ and $\tilde F$ are required to have  $\bigotimes_{x \in \mathbb{Z}^n} \ket{q} \in \mathcal{H}$ as an eigenvector by Lemma~\ref{Rinvariant} (Appendix~\ref{backres}). Thus,  it is not a QLGA in the sense of Definition~\ref{qlgadef}.
\end{proof}

Let us understand the contents of this derivation by way of the proof of Theorem~\ref{ESthm}.  Instead of treating the general case, we take  a simple special case that illustrates the ideas in Theorem~\ref{ESthm} for the QLGA of Definition~\ref{qlgadef}. We take  a one-dimensional QLGA with  neighborhood $\{0,1\}$. Let  $D_0=\emptyset$ and  $D_k=\{0,\ldots,k-1\}$ for $k > 0$, in the ascending chain of sets~\eqref{Dkchain}. Assume that the cell Hilbert space is $W=\mathbb{C}^2\otimes\mathbb{C}^2$, with basis $\{\ket{ij}\}$, $i,j \in \{0,1\}$. Let the quiescent symbol be $\ket{q}=\ket{00}$. Assume also that there is no scattering. Then, in the Heisenberg template, the automorphism $\gamma$ is defined by

\begin{equation*}
\ldots   (\mathbb{I}\otimes \mathbb{I}) \otimes  (\mathbb{I}\otimes \mathbb{I}) \otimes \underbrace{(a \otimes b)}_{\text{$x$}} \otimes (\mathbb{I}\otimes \mathbb{I}) \otimes  (\mathbb{I}\otimes \mathbb{I})  \ldots \xmapsto{\gamma} \ldots   (\mathbb{I}\otimes \mathbb{I}) \otimes  (\mathbb{I}\otimes \mathbb{I})  \otimes \underbrace{(a \otimes \mathbb{I})}_{\text{$x$}} \otimes   \underbrace{(\mathbb{I} \otimes b)}_{\text{$x+1$}}   \otimes (\mathbb{I}\otimes \mathbb{I}) \otimes  (\mathbb{I}\otimes \mathbb{I})  \ldots.
\end{equation*}

Now consider the Schr\"{o}dinger template that  Theorem~\ref{ESthm} implies. We show some of the basis elements of $\mathcal{H}$ and their  transformation by $R$~\eqref{RdefH}.
\begin{align*}
\ldots   \ket{00}  \ket{00} \underbrace{\ket{i0}}_{\text{$x$}}  \underbrace{\ket{0j}}_{\text{$x+1$}}     \ket{00}  \ket{00} \ldots &\xmapsto{R} \ldots   \ket{00}  \ket{00} \underbrace{\ket{ij}}_{\text{$x$}}     \ket{00}  \ket{00} \ldots, \\
\ldots   \ket{00}  \ket{00} \underbrace{\ket{i0}}_{\text{$x$}}  \underbrace{\ket{kj}}_{\text{$x+1$}}   \underbrace{\ket{0l}}_{\text{$x+2$}}   \ket{00}  \ket{00} \ldots &\xmapsto{R} \ldots   \ket{00}  \ket{00} \underbrace{\ket{ij}}_{\text{$x$}} \underbrace{\ket{kl}}_{\text{$x+1$}}      \ket{00}  \ket{00} \ldots
\end{align*}
The intertwining of the representations ($\pi \gamma, \mathcal{H}$) and ($\pi, \mathcal{H}$) by $R$ is clear. We observe that $R$ behaves as the propagation operator $\sigma$ in Definition~\ref{qlgadef}, except that we interpret it as the isomorphism arising from Theorem~\ref{ESthm}. This also provides a concrete interpretation of  Theorem~\ref{ESthm} when  an explicit form of the intertwining operator $R$, the global evolution operator,  can be constructed.   

We can also consider how the general case of Theorem~\ref{ESthm} compares with the QLGA definition,  Definition~\ref{qlgadef}. The local condition in Theorem~\ref{thmS} on the Heisenberg template has the  direct consequence that the cell Hilbert space has a tensor product structure. This is a condition that the Heisenberg template requires and would carry over to the Schr\"{o}dinger  template in Theorem~\ref{ESthm}. In order to create the QLGA of Definition~\ref{qlgadef}, further restrictions were imposed: the associated local scattering operator $F$ associated with $\gamma$ has at least one  product eigenvector, and the quiescent state $\ket{q}$ is one of these eigenvectors.  Theorem~\ref{ESthm} circumvents these at the  expense of the form  of the global evolution operator.

\section{Conclusion} \label{concl}
The  discussion in this paper is focused on the relation between the Schr\"{o}dinger and Heisenberg templates of QCA, within  a borader context of representations of von Neumann algebras.  ITPA is encountered in discussions of  hyperfinite $\Rmnum{2}_1$ factor, such as in~\cite{br:oaqsm1,js:its}. In this paper  the representations of ITPA being alluded to are the $^*$-algebra representations in the specific context of QCA, which in the Schr\"{o}dinger template are defined on a separable Hilbert space, constituting a  $\Rmnum{1}_\infty$ factor. We show that each Heisenberg template has an equivalent Schr\"{o}dinger template. Moreover, an equivalent  Schr\"{o}dinger template exists for any UFC Hilbert space in whose space of  bounded linear operators  the ITPA of the Heisenberg template is  naturally embedded.

 Revisiting the  case of QLGA, a subclass of QCA studied in~\cite{sl:wqcqlg}, we found that the definition of QLGA in~\cite{sl:wqcqlg} was explicit but  restrictive in the sense of the framework and results in this paper. That is precisely because the characterizing local condition of a QLGA belongs to the Heisenberg template. A priori, it does not embody those restrictions, which are imposed through the need to conform with the QLGA definition of~\cite{sl:wqcqlg}, given as a particular Schr\"{o}dinger template.  

Theorem~\ref{ESthm},  does  not provide the  explicit form of a Schr\"{o}dinger template for a QCA  given in the  Heisenberg template.  Nevertheless, it aids in the pursuit of devising QCA  through a Heisenberg template by assuring us of the existence of the  Schr\"{o}dinger template. Our result  extends the reach of  QCA  further in the realms of  quantum computation, information, and simulation, and enhances their versatility as  models of quantum field theories  and quantum gravity.


\section*{Acknowledgments} 
The author would like to thank  Nolan Wallach, David Meyer and members of his research group,  and  Ji\v{r}\'{\i} Lebl for productive and edifying discussions.


\def\MR#1{\relax\ifhmode\unskip\spacefactor3000 \space\fi
  \href{http://www.ams.org/mathscinet-getitem?mr=#1}{MR#1}}

\appendix
\section{Relevant Representation Theory} \label{apprt}

The material on representation theory is adapted from  Goodman and Wallach~\cite{wall:sri}. 

The following result is  Scholium $3.3.2$ in Goodman and Wallach~\cite{wall:ricg} pg. $135$.
\begin{scho}\label{scholium}
Let $\phi \in \begin{rm}{Aut(End}(V))\end{rm}$, where $V$ is a finite dimensional vector space.  Then there exists $G\in \begin{rm}{GL}(V)\end{rm}$ such that $\phi(x)=gxg^{-1}$ for all $x\in  \begin{rm}{End}(V)\end{rm}$.
\end{scho}

Let  $(\gamma, V)$, $(\mu, W)$ be two  representations of an associative algebra  $\mathcal{A}$.  Let $\begin{rm}{Hom}\end{rm}(V,W)$ be the space of $\mathbb{C}$-linear maps from $V$ to $W$. Denote by  $\begin{rm}{Hom}_{\mathcal{A}}\end{rm}(V,W)$ the set of all $T \in \begin{rm}{Hom}\end{rm}(V,W)$ such that $T \gamma(a) = \mu(a)T$ for all $a \in \mathcal{A}$. Such a map is called an \textit{intertwining operator} between  the two representations. Two representations $(\gamma, V)$, $(\mu, W)$ are \textit{equivalent} if there exists an invertible intertwining operator between the two representations.

\begin{lemma}[Schur's Lemma, Lemma $4.1.4$, pg. 180,  Goodman and Wallach~\cite{wall:sri}]  \label{srlem}
Let $(\gamma,V)$ and $(\mu, W)$ be irreducible representations of an associative algebra $\mathcal{A}$. Assume that $V$ and $W$ have at most countable dimensions over $\mathbb{C}$. Then 
\begin{equation*}
\begin{rm}{dim  }\:\end{rm}  \begin{rm}{Hom}_{\mathcal{A}}\end{rm}(V,W) = \left\{
\begin{array}{lr} 
1, & \text{if } (\gamma,V) \cong (\mu, W)\\
0, & \text{otherwise}
\end{array}
\right .
\end{equation*}
\end{lemma}

 Let $V$ be a  vector space. The dimension of $V$ is not assumed to be finite. Let $\mathcal{A}$ be a  finite dimensional semisimple associative subalgebra, and let $\widehat{\mathcal{A}}$ be the set of all equivalence classes of finite dimensional irreducible representations of $\mathcal{A}$.  
 Let $\widehat{\mathcal{A}}$ be the set of all equivalence classes of finite dimensional irreducible representations of $\mathcal{A}$. For each $\lambda \in \widehat{\mathcal{A}}$, fix a module  $(\pi^\lambda, F^\lambda)$.  

Define the map:
\begin{align} \label{Slambda}
s_\lambda : \begin{rm}{Hom}_{\mathcal{A}}\end{rm}(F^\lambda,V) \otimes F^\lambda &\longrightarrow V \\
 u \otimes  v &\mapsto  u( v) \nonumber
\end{align}
Then $s_\lambda$ is an intertwining operator with $ \begin{rm}{Hom}_{\mathcal{A}}\end{rm}(F^\lambda,V) \otimes F^\lambda$ an $\mathcal{A}$-module under the action $a.(u \otimes v) = u \otimes (av)$ for $a \in \mathcal{A}$. 

For $\lambda \in \widehat{\mathcal{A}}$, define the $\lambda$-\textit{isotypic component}:
\begin{align*} 
V_{(\lambda)} :\sum_{W \subset V :  W\sim \lambda }W 
\end{align*}

\begin{defn} \label{localred}
The $A$-module $V$ is locally completely reducible if the cyclic $A$-
submodule $Av$ is finite dimensional and completely reducible for every $v \in V$ .
\end{defn}

\begin{prop} \label{proppdec}
Let $V$ be a locally completely reducible $A$-module. Then the
map $s_\lambda$ gives an $A$-module isomorphism 
\begin{align*} 
\begin{rm}{Hom}_{\mathcal{A}}\end{rm}(F^\lambda,V) \otimes F^\lambda  &\cong V_{(\lambda)}
 \end{align*}
Furthermore,
\begin{align*} 
V = \bigoplus_{\lambda \in \widehat{\mathcal{A}}} V_{(\lambda)} \text{   (Algebraic direct sum)}
\end{align*}
\end{prop}
This is called the \textit{primary decomposition} of $V$.

Let $\mathcal{R} \subset \text{End}(V)$ be  a  subalgebra that acts irreducibly on $V$. Let $\mathcal{A} \subset \mathcal{R}$ be a subalgebra,  acting locally completely reducibly on $V$. Let $\begin{rm}{Spec}\end{rm}(\mathcal{A})$ be the set of irreducible representations of $\mathcal{A}$ that occur in the primary decomposition of $V$. For each $\lambda \in \begin{rm}{Spec}\end{rm}(\mathcal{A})$, fix a module  $(\pi^\lambda, F^\lambda)$. 

Let $\mathcal{R}^\mathcal{A} =\{T \in \mathcal{R} : aT=Ta \text{ for all } a \in \mathcal{A}\}$.

Let
$E^\lambda = \begin{rm}{Hom}_{\mathcal{A}}\end{rm}(F^\lambda,V)$ for $\lambda \in \begin{rm}{Spec}\end{rm}(\mathcal{A})$.
Then $E^\lambda$ is a module for $\mathcal{R}^\mathcal{A}$ in a natural way by left multiplication, since
\begin{equation*}
Tu(\pi^\lambda (a)v) = T (a (u(v))) = a(Tu(v))
\end{equation*}
for $T \in \mathcal{R}^\mathcal{A}$ , $u \in E^\lambda$ , $a \in \mathcal{A}$ , and $v \in F^\lambda$ . As $\mathcal{R}^\mathcal{A}$ commutes with $\mathcal{A}$, there is a representation of $\mathcal{R}^\mathcal{A}\otimes \mathcal{A}$ on $V$. As a module for the algebra $\mathcal{R}^\mathcal{A}\otimes\mathcal{A}$
the space $V$ decomposes, by Proposition~\ref{proppdec},  as
\begin{equation} \label{pdec}
V \cong \bigoplus_{\lambda \in \begin{rm}{Spec}\end{rm}(\mathcal{A})}
E^\lambda \otimes F^\lambda
\end{equation}
In Propostion~\ref{proppdec} an operator $T \in \mathcal{R}^\mathcal{A}$ acts by $T \otimes I$ on the summand of type $\lambda$. $E^\lambda$ is called the {\it multiplicity space} of $\lambda$.

We state  the semisimple algebra version of the Double Commutant Theorem, in Goodman and Wallach~\cite{wall:ricg} pg. $137$.  Let $V$ be a finite dimensional vector space.

\begin{thm}[Double Commutant Theorem for Semisimple Algebras]\label{dblctthm}
 
Suppose  $\mathcal{A} \subset \begin{rm}{End}(V)\end{rm}$   is a semisimple subalgebra  containing the identity operator. 
Then the algebra  $\mathcal{B} =  \begin{rm}{Comm}\end{rm}(\mathcal{A})$  is semisimple and $\mathcal{A} =  \begin{rm}{Comm}\end{rm}(\mathcal{B})$. Furthermore, there exists an $\mathcal{A}$-module isomorphism:
\begin{align*}  
S_\mathcal{A} : \bigoplus^r_{j=1}U_j  \otimes V_j  & \longrightarrow  V  \\
 \sum^r_{j=1}  u_j \otimes  v_j &\mapsto  \sum^r_{j=1}  u_j( v_j) \nonumber
\end{align*}
where $V_j$ is an irreducible $\mathcal{A}$-module, and $U_j =  \begin{rm}{Hom}_{\mathcal{A}}\end{rm}(V_j,V)$.  Under this isomorphism:
\begin{equation*}
\mathcal{A} \cong  \bigoplus^r_{j=1} \mathbb{I}_{U_j}  \otimes \begin{rm}{End}\end{rm}(V_j)
\end{equation*}
and:
\begin{equation*}
\mathcal{B} \cong  \bigoplus^r_{j=1}   \begin{rm}{End}\end{rm}(U_j)  \otimes \mathbb{I}_{V_j}
\end{equation*}
\end{thm}

\subsection*{Related Results from  Representations of Self-Adjoint Algebras} 
 We state a few  results on representations of self-adjoint algebras. Let $V$ be a finite dimensional Hilbert space, with an inner product $\langle \cdot |  \cdot \rangle$, and   $\mathcal{A} \subset \begin{rm}{End}\end{rm}(V)$ a subalgebra. Let the adjoint of an  algebra $\mathcal{A} \subset \begin{rm}{End}\end{rm}(V)$ be $\mathcal{A}^\dag = \{A^\dag : A \in \mathcal{A}\}$. An algebra $\mathcal{A}$  is self-adjoint if $\mathcal{A} = \mathcal{A}^\dag$. From the results in this section, Lemma~\ref{lemmaSA}, and Proposition~\ref{propSA} appear in~\cite{cp:lg} (pg. 145).
\begin{lemma} \label{lemmaSA}
 Let $\mathcal{A} \subset \begin{rm}{End}\end{rm}(V)$ be a self-adjoint subalgebra. If $W \in V$ is an $\mathcal{A}$-invariant subspace, then $W^\perp = \{ {v} \in V : \langle v | w \rangle = 0 \: \forall \:  {w} \in W\}$ is $\mathcal{A}$-invariant. 
\end{lemma}

\begin{proof}
Let $ {w} \in W$, $ {v} \in  W^\perp$, $A \in \mathcal{A}$. Then $A^\dag \in \mathcal{A}$, which implies $A^\dag  {w} \in W$ $\implies$   $\langle A  {v} | w \rangle  = \langle  v | A^\dag w \rangle = 0$ $\implies$  $A  {v} \in W^\perp$.
\end{proof}

\begin{prop} \label{propSA}
 Let $\mathcal{A} \subset \begin{rm}{End}\end{rm}(V)$ be a self-adjoint subalgebra. Then $V$ is an orthogonal direct sum of irreducible  $\mathcal{A}$-modules. In particular, $V$ is a completely reducible $\mathcal{A}$-module.
\end{prop}

\begin{proof}
Let   $W \subset V$ be an  $\mathcal{A}$-invariant subspace of minimal dimension.  Then it is by definition irreducible. Since $ \mathcal{A} =  \mathcal{A}^\dag$, by Lemma~\ref{lemmaSA},  $V = W \oplus W^\perp$ is an orthogonal direct sum  of $\mathcal{A}$-modules. The conclusion follows by  induction on dimension.
\end{proof}

\begin{cor} \label{sass}
 Let $\mathcal{A} \subset \begin{rm}{End}\end{rm}(V)$ be a self-adjoint subalgebra  containing the identity operator. Then $\mathcal{A}$ is semisimple.
\end{cor}
\begin{proof}
By Proposition~\ref{propSA}, $V$ is a completely reducible $\mathcal{A}$-module. Then Proposition~\ref{proppdec} implies that $\mathcal{A}$ is semisimple.
\end{proof}

\begin{cor}  \label{corSSB}
Let $\mathcal{A} \subset \begin{rm}{End}\end{rm}(V)$ be a self-adjoint subalgebra.  Then there exist for every $\lambda  \in  \widehat{\mathcal{A}}$,  a set of intertwining  operators $\{u^\lambda_j\} \subset  U^\lambda$, such that the $\lambda$-isotypic component can be written as an orthogonal direct sum: $V^{[\lambda]} =  \bigoplus_{j} u^\lambda_j(V^\lambda)$. 
\end{cor}

 \begin{proof}
By Proposition~\ref{propSA}, $V$ is a completely reducible $\mathcal{A}$-module. Thus by Proposition~\eqref{proppdec}, $s_\lambda$ in~\eqref{Slambda} is an  $\mathcal{A}$-module  isomorphism:
\begin{align*} 
s_\lambda: U^{\lambda} \otimes V^\lambda  \longrightarrow  V^{[\lambda]}
\end{align*}
Again by  Proposition~\ref{propSA}, the $\lambda$-isotypic component,
\begin{align*} 
V^{[\lambda]} = \bigoplus_{W_j \subset V :  W_j\sim \lambda }W_j 
\end{align*}
is an orthogonal direct sum.  Choose $u_j$ such that $u_j(V^\lambda) = W_j$.
\end{proof}

\subsection*{Outer Tensor  Product Representation}

Let $\mathcal{A},\mathcal{B}$ be associative algebras. Let $(\sigma, V)$ and $(\tau, W)$ be finite dimensional representations of $\mathcal{A}$, $\mathcal{B}$ respectively. The {\it outer tensor product}  representation $(\sigma \widehat{\otimes} \tau, V \otimes W)$ of the tensor product $\mathcal{A} \otimes \mathcal{B}$ is defined as follows. On a basis element $a \otimes b \in \mathcal{A} \otimes \mathcal{B}$:
\begin{equation} \label{outtenrep}
\big(\sigma \widehat{\otimes} \tau \big) (a \otimes b) = \sigma(a) \otimes \tau(b)
\end{equation}
Extend by linearity to a representation of $\mathcal{A} \otimes \mathcal{B}$. If $(\sigma, V)$ and $(\tau, W)$ are irreducible representations, then $(\sigma \widehat{\otimes} \tau, V \otimes W)$ is an irreducible representation.

\section{Relevant Direct Limit Background} \label{dirlimback}
The reader is referred to Rotman~\cite{rot:hom} for background material on direct limits. Let $(I,\le)$ be a {\it directed}  set, i.e., a set with a preorder $\le$ (partial order is a special case), with the following property: for any pair of elements  $a,b \in I$, there exists an element $c \in I$  such that  $a \le c$ and $b \le c$.

Let $\mathcal{C}$ be a category. Let  a family of objects in $\mathcal{C}$ be indexed by $I$: $\{C_i\}$, $i \in I$. Let $\{f^i_j\}$, $i,j \in I$, $i\le j$,  be a {\it directed family  of morphisms}, by which we mean  $f^i_j \in \text{Mor}(C_i, C_j)$, such that $f^i_i = \text{id}$, and $f^j_k \circ f^i_j = f^i_k$ for $i\le j\le k$. Then the pair $\big((C_i),(f^i_j)\big)$, abbreviated $\{C_i,f^i_j\}$ is called a {\it direct system} (over $I$).

Consider the following category of  pairs, each such pair being a {\it target}: $(C,\{f^i\}_{i\in I})$, where $C\in \mathcal{C}$ and each $f^i$ is an {\it insertion} morphism $f^i: C_i \rightarrow C$, such that $f^i=f^j \circ f^i_j$ for $i\le j$. 

A target $(A,\{h^i\}_{i\in I})$ is a direct limit, denoted 
$\underrightarrow{A_i}$, 
if it satisfies the following {\it universal property}: for every target  $(G,\{g^i\}_{i\in I})$  there exists a unique morphism $\phi: A \rightarrow G$, such that $g^i=\phi \circ h^i$.  Direct limit may not exist, but  if it does,  it is unique up to  unique isomorphism, i.e., if $A$ is a direct limit, then for any other direct limit $A'$, there exists a unique isomorphism $A' \rightarrow A$ that commutes with the insertion morphisms.

Let $R$ be a ring.
\begin{thm} [Proposition $5.23$, Rotman~\cite{rot:hom}] \label{dirlim} 
The direct limit of any direct system $(M_i , f^i_j)$ of left $R$-
modules over a partially ordered index set $I$ exists.
\end{thm}

\section{Relevant QCA Related Results} \label{backres}

\subsection*{Hilbert space of finite, unbounded configurations as a direct limit}
For any subset $D \subset \mathbb{Z}^n$, let 
\begin{equation*}
W_{D} = \bigotimes_{x \in D} W \:.
\end{equation*}
Consider the ascending chain of subsets~\eqref{Dkchain}.  For $D_m \subset D_n$, define a linear, isometric map
\begin{equation*}
f_{m,n}: W_{D_m} \rightarrow W_{D_n}
\end{equation*}
by first expressing
\begin{equation*}
W_{D_n} = \bigotimes_{x \in D_m} W  \otimes \bigotimes_{y \in D_n \setminus D_m} W   \:.
\end{equation*}
and then defining $f_{m,n}$ as
\begin{equation*}
f_{m,n}(\ket{v}) =  \ket{v} \otimes \bigotimes_{y \in D_n \setminus D_m} \ket{q},\quad \forall \ket{v} \in W_{D_m}\:.
\end{equation*}
For $D_l \subset D_m \subset D_n$,  $f_{n,l} = f_{m,n} \circ f_{l,m}$. Thus $\{W_{D_m}, f_{m,n} \}$ form  a direct system (see Appendix~\ref{dirlimback}). Then the direct limit of this system  $\underrightarrow{\lim} W_{D_m}$ exists by Theorem~\ref{dirlim},  with a pre-Hilbert structure induced from the inner product on $W$.
Let the  canonical insertion map for  $W_{D_m}$  be $h_{m}: W_{D_m} \rightarrow \underrightarrow{\lim} W_{D_m}$. Under  $h_{m}$, $W_{D_m}$ can be considered a subspace of $\underrightarrow{\lim} W_{D_m}$, so $\underrightarrow{\lim} W_{D_m}$ is the  union of $\{ W_{D_m}\}$.  Thus we see that 
\begin{equation} \label{Vaslim}
\mathbb{V} = \underrightarrow{\lim} W_{D_m} \:.
\end{equation}
 The {\it Hilbert space of finite, unbounded configurations} (UFC Hilbert space), denoted by $\mathcal{H}$, is the completion of $\mathbb{V}$ under the inner product norm (the  pre-Hilbert structure)  on ${\mathbb{V}}$ induced  from the inner product on $W$.

\begin{remark}
For the direct system above, we do not need to use an ascending chain of sets~\eqref{Dkchain} as the directed set on which the system is defined. We can simply use set inclusion as a partial order on finite subsets of $\mathbb{Z}^n$ to form the direct system. This is true for  all the direct systems and limits in this paper.
\end{remark}

\subsection*{$\bigotimes_{x \in \mathbb{Z}^n} \ket{q}$ as an eigenvector of translation-invariant unitary operator}
 
 \begin{lemma}[Lemma  $\rm{III}.5$ in~\cite{sl:wqcqlg}] \label{Rinvariant}
An invertible  and translation invariant operator $M$ on a Hilbert space of finite configurations $\mathcal{H}$ has $\bigotimes_{x \in \mathbb{Z}^n} \ket{q}$ as an eigenvector:
\begin{equation*}
M\big(\bigotimes_{x \in \mathbb{Z}^n} \ket{q} \big)= \lambda \bigotimes_{x \in \mathbb{Z}^n} \ket{q}  
\end{equation*}
for some $\lambda \in \mathbb{C}^\times$. In particular, if $M$ is unitary and translation invariant, then $\lambda = e^{i \theta}$ for some $\theta \in \mathbb{R}$.
\end{lemma}

\end{document}